\documentclass[reqno]{amsart}
\usepackage{amsmath}
\usepackage{amsfonts}
\usepackage{amssymb}
\usepackage{amssymb,graphics,graphicx,bbm,hyperref,color}

\usepackage{multicol}
\usepackage{enumitem} 

\definecolor{gray}{rgb}{0.93,0.93,0.93}
\definecolor{light-gold}{rgb}{0.99,0.97,0.78}

\def\be{\begin{equation}}
\def\ee{\end{equation}}
\def\bm{\begin{multline}}
\def\bfig{\begin{figure}[htb]}
\def\efig{\end{figure}}
\newcommand{\dd}{{\rm d}}
\newcommand{\e}[1]{\,{\rm e}^{#1}\,}
\newcommand{\ii}{{\rm i}}
\def\Tr{{\operatorname{Tr\,}}}

\newcommand{\sumtwo}[2]{\sum_{\substack{#1 \\ #2}}}

\numberwithin{equation}{section}
\newtheorem{theorem}{Theorem}
\newtheorem{proposition}[theorem]{Proposition}
\newtheorem{lemma}[theorem]{Lemma}
\newtheorem{corollary}[theorem]{Corollary}

\theoremstyle{remark}

\newcommand{\caH}{{\mathcal H}}

\newcommand{\bbC}{{\mathbb C}}

\newcommand{\bbN}{{\mathbb N}}

\newcommand{\bbZ}{{\mathbb Z}}

\newcommand{\eps}{{\varepsilon}}

\def\bbone{{\mathchoice {\rm 1\mskip-4mu l} {\rm 1\mskip-4mu l} {\rm 1\mskip-4.5mu l} {\rm 1\mskip-5mu l}}}

\newcommand{\lda}{\langle\!\langle}
\newcommand{\rda}{\rangle\!\rangle}

\makeatletter
  \def\tagform@#1{\maketag@@@{\footnotesize{(#1)}\@@italiccorr}}
\makeatother

\renewcommand{\eqref}[1]{(\ref{#1})}

\begin{document}


\title{Correlation inequalities for the quantum XY model}

\author{Costanza Benassi}
\author{Benjamin Lees}
\author{Daniel Ueltschi}
\address{Department of Mathematics, University of Warwick,
Coventry, CV4 7AL, United Kingdom}
\email{C.Benassi@warwick.ac.uk}
\email{B.Lees@warwick.ac.uk}
\email{daniel@ueltschi.org}



\begin{abstract}
We show the positivity or negativity of truncated correlation functions in the quantum XY model with spin 1/2 (at any temperature) and spin 1 (in the ground state). These Griffiths-Ginibre inequalities of the second kind generalise an earlier result of Gallavotti.
\end{abstract}

\thanks{\copyright{} 2016 by the authors. This paper may be reproduced, in its entirety, for non-commercial purposes.}

\maketitle

\section{Introduction \& results}
\label{sec intro}

Correlation inequalities, initially proposed by Griffiths \cite{Gri}, have been an invaluable tool in the study of several classical spin systems. They have helped to establish the infinite volume limit of correlation functions, the monotonicity of spontaneous magnetisation, and they have allowed to make comparisons between different spatial or spin dimensions. It would be helpful to have similar tools for the study of quantum systems, but results are scarce. Gallavotti has obtained inequalities for truncated functions in the case of the quantum XY model with spin 1/2 and pair interactions \cite{Gal}. Inequalities for (untruncated) correlations in more general models were proposed in \cite{FU}.

In his extension of Griffiths' inequalities, Ginibre proposed a setting that also applies to quantum spin systems \cite{Gin}. The goal of this note is to show that the quantum XY model fits the setting, at least with $S = \frac12$ and $S=1$. It follows that many truncated correlation functions take a fixed sign.

Let $\Lambda$ denote the (finite) set of sites that host the spins. The Hilbert space of the model is $\caH_{\Lambda} = \otimes_{x\in\Lambda} \bbC^{2S+1}$ with $S \in \frac12 \bbN$. Let $S^{i}$, $i=1,2,3$ denote usual spin operators on $\bbC^{2S+1}$; that is, they satisfy the commutation relations $[S^{1},S^{2}] = \ii S^{3}$, and other relations obtained by cyclic permutation of the indices $1,2,3$. They also satisfy the identity $(S^{1})^{2} + (S^{2})^{2} + (S^{3})^{2} = S(S+1)$.
Finally, let $S_{x}^{i} = S^{i} \otimes \bbone_{\Lambda\setminus\{x\}}$ denote the spin operator at site $x$. We consider the hamiltonian
\be\label{Ham}
H_{\Lambda} = -\sum_{A \subset \Lambda} \Bigl( J_{A}^{1} \prod_{x\in A} S_{x}^{1} + J_{A}^{2} \prod_{x\in A} S_{x}^{2} \Bigr).
\ee
Here, $J_{A}^{i}$ is a nonnegative coupling constant for each subset of $A \subset \Lambda$ and each spin direction $i \in \{1,2\}$. The expected value of an observable $a$ (that is, an operator on $\caH_{\Lambda}$) in the Gibbs state with hamiltonian $H_{\Lambda}$ and at inverse temperature $\beta>0$ is
\be
\label{def Gibbs}
\langle a \rangle = \frac1{Z(\Lambda)} \Tr a \e{-\beta H_{\Lambda}},
\ee
where the normalisation $Z(\Lambda)$ is the partition function
\be
Z(\Lambda) = \Tr \e{-\beta H_{\Lambda}}.
\ee
Traces are taken in $\caH_{\Lambda}$.
We also consider Schwinger functions that are defined for $s\in[0,1]$ by
\be
\label{def Schwinger}
\langle a ; b \rangle_{s} = \frac1{Z(\Lambda)} \Tr a \e{-s\beta H_{\Lambda}} b \e{-(1-s) \beta H_{\Lambda}}.
\ee

Our first result holds for $S=\frac12$ and all temperatures.

\begin{theorem}
\label{thm ineq}
Assume that $J_{A}^{i} \geq 0$ for all $A \subset \Lambda$ and all $i \in \{1,2\}$. Assume also that $S=\frac12$. Then for all $A,B \subset \Lambda$, and all $s\in[0,1]$, we have
\[
\begin{split}
&\Bigl\langle \prod_{x\in A} S_{x}^{1} ; \prod_{x\in B} S_{x}^{1} \Bigr\rangle_{s} - \Bigl\langle \prod_{x\in A} S_{x}^{1} \Bigr\rangle \, \Bigl\langle \prod_{x\in B} S_{x}^{1} \Bigr\rangle \geq 0; \\
&\Bigl\langle \prod_{x\in A} S_{x}^{1} ; \prod_{x\in B} S_{x}^{2} \Bigr\rangle_{s} - \Bigl\langle \prod_{x\in A} S_{x}^{1} \Bigr\rangle \, \Bigl\langle \prod_{x\in B} S_{x}^{2} \Bigr\rangle \leq 0.
\end{split}
\]
\end{theorem}

Clearly, other inequalities can be generated using spin symmetries.
The corresponding inequalities for the classical XY model have been proposed in \cite{KPV}.

The proof of Theorem \ref{thm ineq} can be found in Section \ref{sec S=1/2}. It is based on Ginibre's structure \cite{Gin}. It is simpler than Gallavotti's, who used an ingenious approach based on the Trotter product formula, on a careful analysis of transition operators, and on Griffiths' inequalities for the classical Ising model \cite{Gal}. Our proof allows us to go beyond pair interactions.

A consequence of Theorem \ref{thm ineq} is the monotonicity of certain spin correlations with respect to the coupling constants:

\begin{corollary}
\label{cor ineq}
Under the same assumptions as in the above theorem, we have for all $A,B \subset \Lambda$ that
\[
\begin{split}
&\frac{\partial}{\partial J_{A}^{1}} \Bigl\langle \prod_{x\in B} S_{x}^{1} \Big\rangle \geq 0; \\
&\frac{\partial}{\partial J_{A}^{1}} \Bigl\langle \prod_{x\in B} S_{x}^{2} \Big\rangle \leq 0.
\end{split}
\]
\end{corollary}

The first inequality states that correlations increase when the coupling constants increase (in the same spin direction). The second inequality is perhaps best understood classically; if the first component of the spins increases, the other components must decrease because the total spin is conserved.
Corollary \ref{cor ineq} follows immediately from Theorem \ref{thm ineq} since
\be
\frac1\beta \frac{\partial}{\partial J_{A}^{i}} \Bigl\langle \prod_{x\in B} S_{x}^{j} \Big\rangle = \int_{0}^{1} \Bigl[ \Bigl\langle \prod_{x\in B} S_{x}^{j} ; \prod_{x\in A} S_{x}^{i} \Bigr\rangle_{s} - \Bigl\langle \prod_{x\in B} S_{x}^{j} \Bigr\rangle \, \Bigl\langle \prod_{x\in A} S_{x}^{i} \Bigr\rangle \Bigr] \dd s.
\ee

We use this corollary in Section \ref{sec inf vol limit} to give a partial construction of infinite-volume Gibbs states.

The case of higher spins, $S > \frac12$, is much more challenging, but we have obtained an inequality that is valid in the ground state of the $S=1$ model. Recall that the states $\langle\cdot\rangle$ and $\langle\cdot;\cdot\rangle_s$, defined in Eqs \eqref{def Gibbs} and \eqref{def Schwinger}, depend on the inverse temperature $\beta$.

\begin{theorem}
\label{thm ineq2}
Assume that $J_{A}^{i} \geq 0$ for all $A \subset \Lambda$ and all $i \in \{1,2\}$. Assume also that $S=1$. Then for all $A,B \subset \Lambda$, and all $s\in[0,1]$, we have
\[
\begin{split}
&\lim_{\beta\to\infty} \biggl[ \Bigl\langle \prod_{x\in A} S_{x}^{1} ; \prod_{x\in B} S_{x}^{1} \Bigr\rangle_{s} - \Bigl\langle \prod_{x\in A} S_{x}^{1} \Bigr\rangle \, \Bigl\langle \prod_{x\in B} S_{x}^{1} \Bigr\rangle \biggr] \geq 0; \\
&\lim_{\beta\to\infty} \biggl[ \Bigl\langle \prod_{x\in A} S_{x}^{1} ; \prod_{x\in B} S_{x}^{2} \Bigr\rangle_{s} - \Bigl\langle \prod_{x\in A} S_{x}^{1} \Bigr\rangle \, \Bigl\langle \prod_{x\in B} S_{x}^{2} \Bigr\rangle \biggr] \leq 0.
\end{split}
\]
\end{theorem}

The proof of this theorem can be found in Section \ref{sec S=1}. It uses Theorem \ref{thm ineq}.

\section{Infinite volume limit of correlation functions}
\label{sec inf vol limit}

Infinite volume limits of Gibbs states are notoriously delicate issues; we show in this section that Theorem \ref{thm ineq} (and Corollary \ref{cor ineq}) give partial but useful information: For Gibbs states ``with $+$ boundary conditions", the infinite volume limits of many correlation functions exist.

Let us recall the notion of infinite volume limit. Let $(t_\Lambda)_{\Lambda \subset\subset \bbZ^d}$ be a sequence of real or complex numbers, indexed by finite subsets of $\bbZ^d$. We say that $t_\Lambda \to t$ as $\Lambda \nearrow \bbZ^d$ if
\be
\lim_{n\to\infty} t_{\Lambda_n} = t
\ee
along every sequence $(\Lambda_n)$ of increasing finite subsets that tends to $\bbZ^d$. That is, the sequence satisfies $\Lambda_{n+1} \supset \Lambda_n$, and, for any finite $A \subset\subset \bbZ^d$, there exists $n_A$ such that $\Lambda_n \supset A$ for all $n \geq n_A$.

We assume the interaction is finite-range: There exists $R$ such that $J^i_A=0$ whenever ${\rm diam} \, A > R$. Let $\Lambda_R$ denote the enlarged domain
\be
\Lambda_R = \{ x \in \bbZ^d : {\rm dist}(x,\Lambda) \leq R \}.
\ee
Let $\partial_R \Lambda = \Lambda_R \setminus \Lambda$ be the exterior boundary of $\Lambda$. We consider the hamiltonian $H_{\Lambda_R}^\eta$ with field on the exterior boundary:
\be
H_{\Lambda_R}^\eta = -\sum_{A \subset \Lambda_R} \Bigl( J^1_A \prod_{x\in A} S_x^1 + J^2_A \prod_{x\in A} S_x^2 \Bigr) - \eta \sum_{x \in \partial_R \Lambda} S_x^1.
\ee
Temperature does not play a r\^ole in this section so we set $\beta=1$.
The relevant (finite volume) Gibbs state is the linear functional that, to any operator $a$ on $\caH_\Lambda$, assigns the value
\be
\label{Gibbs +}
\langle a \rangle_{\Lambda}^{(+)} = \lim_{\eta\to\infty} \frac{\Tr a \e{-H_{\Lambda_R}^\eta}}{\Tr \e{-H_{\Lambda_R}^\eta}}.
\ee
Traces are taken in $\caH_{\Lambda_R}$ (and $a$ on $\caH_\Lambda$ is identified with $a \otimes \bbone_{\partial_R \Lambda}$ on $\caH_{\Lambda_R}$). We comment below on the relevance of this definition for Gibbs states. But first, we observe that the limit $\eta\to\infty$ exists.

\begin{proposition}
\label{prop H_+}
For all operators $a$ on $\caH_\Lambda$, the limit in \eqref{Gibbs +} exists and is equal to
\[
\langle a \rangle_{\Lambda}^{(+)} = \frac{\Tr a \e{-H_{\Lambda}^{(+)}}}{\Tr \e{-H_{\Lambda}^{(+)}}},
\]
where traces are taken in $\caH_\Lambda$ and
\[
H_{\Lambda}^{(+)} = -\sum_{A \subset \Lambda} \Bigl( J^1_A \prod_{x\in A} S_x^1 + J^2_A \prod_{x\in A} S_x^2 \Bigr) - \sumtwo{A \subset \Lambda_R}{A \cap \partial_R \Lambda \neq \emptyset} 2^{-|A \cap \partial_R \Lambda|} J_A^1 \prod_{x\in A \cap \Lambda} S_x^1.
\]
\end{proposition}

\begin{proof}
We can add a convenient constant to the hamiltonian without changing the corresponding Gibbs state, so we consider
\be
\Tr a \exp\biggl\{ \sum_{A \subset \Lambda_R} \Bigl( J^1_A \prod_{x\in A} S_x^1 + J^2_A \prod_{x\in A} S_x^2 \Bigr) + \eta \sum_{x \in \partial_R \Lambda} (S_x^1 - \tfrac12) \biggr\}.
\ee
We have
\be
\lim_{\eta\to\infty} \e{\eta (S_x^1-\frac12)} = P_x^+,
\ee
where $P_x^+$ is the projector onto the eigenstates of $S_x^1$ with eigenvalue $\frac12$. Writing $P_A^+ = \prod_{x\in A} P_x^+$, we have
\be
\begin{split}
&P_{\partial_R \Lambda}^+ \Bigl( \prod_{x \in A} S_x^1 \Bigr) P_{\partial_R \Lambda}^+ = 2^{-|A \cap \partial_R \Lambda|} \Bigl( \prod_{x \in A \cap \Lambda} S_x^1 \Bigr) P_{\partial_R \Lambda}^+, \\
&P_{\partial_R \Lambda}^+ \Bigl( \prod_{x \in A} S_x^2 \Bigr) P_{\partial_R \Lambda}^+ = 0 \qquad \text{if } A \cap \partial_R \Lambda \neq \emptyset.
\end{split}
\ee
Then, since the Trotter expansion converges uniformly in $\eta$, we have
\be
\begin{split}
&\lim_{\eta\to\infty} \Tr a\exp\biggl\{ \sum_{A \subset \Lambda_R} \Bigl( J^1_A \prod_{x\in A} S_x^1 + J^2_A \prod_{x\in A} S_x^2 \Bigr) + \eta \sum_{x \in \partial_R \Lambda} (S_x^1 - \tfrac12) \biggr\} \\
&= \lim_{n\to\infty} \lim_{\eta\to\infty} \Tr a \biggl[ \biggl( 1 + \tfrac1n \sum_{A \subset \Lambda_R} \Bigl( J^1_A \prod_{x\in A} S_x^1 + J^2_A \prod_{x\in A} S_x^2 \Bigr) \biggr) \e{\frac1n \eta \sum_{x \in \partial_R \Lambda} (S_x^1 - \frac12)} \biggr]^n \\
&= \lim_{n\to\infty} \Tr a \bigl[ 1 - \tfrac1n H_{\Lambda}^{(+)} \bigr]^n \\
&= \Tr a \e{-H_{\Lambda}^{(+)}}.
\end{split}
\ee
\end{proof}

The challenge is to prove that $\langle a \rangle_{\Lambda}^{(+)}$ converges as $\Lambda \nearrow \bbZ^d$, for any operator $a$ on $\caH_{\Lambda'}$ with $\Lambda' \subset\subset \bbZ^d$ (again, $a$ on $\caH_{\Lambda'}$ is identified with $a \otimes \bbone_{\Lambda \setminus \Lambda'}$ on $\caH_\Lambda$ with $\Lambda \supset \Lambda'$). We can use the correlation inequalities to establish the existence of the infinite volume limit for certain operators $a$.

\begin{theorem}
\label{thm inf vol limit}
For every finite $A \subset\subset \bbZ^d$ and every $i \in \{1,2\}$, $\bigl\langle \prod_{x\in A} S_x^i \bigr\rangle_{\Lambda}^{(+)}$ converges as $\Lambda \nearrow \bbZ^d$.
\end{theorem}

\begin{proof}
If $\Lambda \subset \Lambda'$, let us define the hamiltonian
\be
H_{\Lambda,\Lambda_R'}^\eta = -\sum_{A \subset \Lambda_R'} \Bigl( J^1_A \prod_{x\in A} S_x^1 + J^2_A \prod_{x\in A} S_x^2 \Bigr) - \eta \sum_{x \in \Lambda_R' \setminus \Lambda} S_x^1.
\ee
Adapting the proof of Proposition \ref{prop H_+}, we can check that we have, for all operators on $\caH_\Lambda$,
\be
\langle a \rangle_{\Lambda}^{(+)} = \lim_{\eta\to\infty} \frac{\Tr a \e{-H_{\Lambda,\Lambda_R'}^\eta}}{\Tr \e{-H_{\Lambda,\Lambda_R'}^\eta}},
\ee
where traces are taken in $\caH_{\Lambda_R'}$. Corollary \ref{cor ineq} implies that
\be
\frac{\Tr \bigl( \prod_{x \in A} S_x^1 \bigr) \e{-H_{\Lambda,\Lambda_R'}^\eta}}{\Tr \e{-H_{\Lambda,\Lambda_R'}^\eta}} \geq \frac{\Tr \bigl( \prod_{x \in A} S_x^1 \bigr) \e{-H_{\Lambda_R'}^\eta}}{\Tr \e{-H_{\Lambda_R'}^\eta}} ;
\ee
the opposite inequality holds when $\prod S_x^1$ is replaced by $\prod S_x^2$. Thus $\langle \prod_{x\in A} S_x^i \rangle_{\Lambda}^{(+)}$ is monotone decreasing for $i=1$, and monotone increasing for $i=2$. It is also bounded, so it converges.
\end{proof}

Finally, let us comment on the relevance of this Gibbs state with $+$ boundary conditions. Consider the case of the isotropic XY model, where $J_A^1 = J_A^2$ for all $A \subset\subset \bbZ^d$. At low temperatures, the infinite volume state $\langle \cdot \rangle^{(+)} = \lim_{\Lambda \nearrow \bbZ^d} \langle \cdot \rangle_{\Lambda}^{(+)}$ is expected to be extremal and to describe a system with spontaneous magnetisation in the direction 1 of the spins. One can apply rotations in the 1-2 plane to get all other (translation-invariant) extremal Gibbs states. Much work remains to be done to make this rigorous, but Theorem \ref{thm inf vol limit} seems to be a useful step.

\section{The case $S=\frac12$}
\label{sec S=1/2}

We can define the spin operators as $S^{i} = \frac12 \sigma^{i}$, where the $\sigma^{i}$s are the Pauli matrices
\be
\label{def Pauli}
\sigma^{1} = \left( \begin{matrix} 0 & 1 \\ 1 & 0 \end{matrix} \right), \qquad \sigma^{2} = \left( \begin{matrix} 0 & -\ii \\ \ii & 0 \end{matrix} \right), \qquad \sigma^{3} = \left( \begin{matrix} 1 & 0 \\ 0 & -1 \end{matrix} \right).
\ee
It is convenient to work with the hamiltonian with interactions in the 1-3 spin directions, namely
\be
\label{def ham 1-3}
H_{\Lambda} = -\sum_{A \subset \Lambda} \Bigl( J_{A}^{1} \prod_{x\in A} S_{x}^{1} + J_{A}^{3} \prod_{x\in A} S_{x}^{3} \Bigr).
\ee
Following Ginibre \cite{Gin}, we introduce the product space $\caH_{\Lambda} \otimes \caH_{\Lambda}$. 
Given an operator $a$ on $\caH_{\Lambda}$, we consider the operators $a_{+}$ and $a_{-}$ on the product space, defined by
\be
a_{\pm} = a \otimes \bbone \pm \bbone \otimes a.
\ee
The Gibbs state in the product space is
\be
\lda \cdot \rda = \frac1{Z(\Lambda)^{2}} \Tr \cdot \e{-H_{\Lambda,+}},
\ee
where $H_{\Lambda,+} = H_{\Lambda} \otimes \bbone + \bbone \otimes H_{\Lambda}$.
Without loss of generality, we set $\beta=1$ in this section.
We also need the Schwinger functions in the product space, namely
\be
\lda \cdot ; \cdot \rda_{s} = \frac1{Z(\Lambda)^{2}} \Tr \cdot \e{-sH_{\Lambda,+}} \cdot \e{-(1-s)H_{\Lambda,+}}.
\ee

\begin{lemma}
\label{lem correlation}
For all observables $a,b$ on $\caH_{\Lambda}$, we have
\[
\begin{split}
&\langle ab \rangle - \langle a \rangle \langle b \rangle = \tfrac12 \lda a_{-} b_{-} \rda, \\
&\langle a; b \rangle_{s} - \langle a \rangle \langle b \rangle = \tfrac12 \lda a_{-} ; b_{-} \rda_{s}.
\end{split}
\]
\end{lemma}

\begin{proof}
It is enough to prove the second line. The right side is equal to
\be
\begin{split}
\lda a_{-} ; b_{-} \rda_{s} = \frac1{Z(\Lambda)^{2}} \Bigl[ &\Tr (a \otimes \bbone) \e{-sH_{\Lambda,+}} (b \otimes \bbone) \e{-(1-s) H_{\Lambda,+}} \\
&\qquad\qquad + \Tr (\bbone \otimes a) \e{-sH_{\Lambda,+}} (\bbone \otimes b) \e{-(1-s) H_{\Lambda,+}} \\
&- \Tr (\bbone \otimes a) \e{-sH_{\Lambda,+}} (b \otimes \bbone) \e{-(1-s) H_{\Lambda,+}} \\
&\qquad\qquad - \Tr (a \otimes \bbone) \e{-sH_{\Lambda,+}} (\bbone \otimes b) \e{-(1-s) H_{\Lambda,+}} \Bigr].
\end{split}
\ee
The first two lines of the right side give $2 \langle a ; b \rangle_{s}$ and the last two lines give $2 \langle a \rangle \langle b \rangle$.
\end{proof}

Next, a simple lemma with a useful formula.

\begin{lemma}
\label{lem useful}
For all operators $a,b$ on $\caH_{\Lambda}$, we have
\[
(ab)_{\pm} = \tfrac12 a_{+} b_{\pm} + \tfrac12 a_{-} b_{\mp}.
\]
\end{lemma}

The proof is straightforward algebra. Notice that both terms of the right side have {\it positive} factors. Now comes the key observation that leads to positive (and negative) correlations.

\begin{lemma}
\label{lem key}
There exists an orthonormal basis on $\bbC^{2} \otimes \bbC^{2}$ such that $S_{+}^{1}, S_{-}^{1}, S_{+}^{3}, -S_{-}^{3}$ have nonnegative matrix elements.
\end{lemma}

As a consequence, there exists an orthonormal basis of $\caH_{\Lambda} \otimes \caH_{\Lambda}$ such that $S_{x,+}^{1}, S_{x,-}^{1}, S_{x,+}^{3}$, and $-S_{x,-}^{3}$ have nonnegative matrix elements.

\begin{proof}[Proof of Lemma \ref{lem key}]
For $\eps_{1}, \eps_{2} = \pm$, let $| \eps_{1}, \eps_{2} \rangle$ denote the eigenvectors of $S^{3} \otimes \bbone$ and $\bbone \otimes S^{3}$ with respective eigenvalues $\frac12 \eps_{1}$ and $\frac12 \eps_{2}$. It is well-known that $S^{1} \otimes \bbone \; | \eps_{1}, \eps_{2} \rangle = \frac12 |-\eps_{1}, \eps_{2} \rangle$ and similarly for $\bbone \otimes S^{1}$. The convenient basis in $\bbC^{2} \otimes \bbC^{2}$ consists of the following four elements:
\be
\begin{aligned}
&p_{+} = \tfrac1{\sqrt2} \bigl( |++\rangle + |--\rangle \bigr), \qquad & q_{+} = \tfrac1{\sqrt2} \bigl( |-+\rangle + |+-\rangle \bigr), \\
&p_{-} = \tfrac1{\sqrt2} \bigl( |++\rangle - |--\rangle \bigr), & q_{-} = \tfrac1{\sqrt2} \bigl( |-+\rangle - |+-\rangle \bigr).
\end{aligned}
\ee
Direct calculations show that
\be
\begin{split}
&(p_{+}, S_{+}^{1} q_{+}) = (q_{+}, S_{+}^{1} p_{+}) = 1, \\
&(p_{-}, S_{-}^{1} q_{-}) = (q_{-}, S_{-}^{1} p_{-}) = 1, \\
&(p_{+}, S_{+}^{3} p_{-}) = (p_{-}, S_{+}^{3} p_{+}) = 1, \\
&(q_{+}, S_{-}^{3} q_{-}) = (q_{-}, S_{-}^{3} q_{+}) = -1.
\end{split}
\ee
All other matrix elements are zero.
\end{proof}

\begin{proof}[Proof of Theorem \ref{thm ineq} for $S=\frac12$]
We use Lemma \ref{lem correlation} in order to get
\be
\Bigl\langle \prod_{x\in A} S_{x}^{1} ; \prod_{x\in B} S_{x}^{1} \Bigr\rangle_{s} - \Bigl\langle \prod_{x\in A} S_{x}^{1} \Bigr\rangle \, \Bigl\langle \prod_{x\in B} S_{x}^{1} \Bigr\rangle = \tfrac12 \Bigl\langle \!\! \Bigl\langle \Bigl( \prod_{x\in A} S_{x}^{1} \Bigr)_{-} ; \Bigl( \prod_{x\in B} S_{x}^{1} \Bigr)_{-} \Bigr\rangle \!\! \Bigr\rangle_{s}.
\ee
In order to make visible the sign of the right side, we expand the exponentials in Taylor series, so as to get a positive linear combination of terms of the form
\be
\Tr_{\caH_\Lambda \otimes \caH_\Lambda} \Bigl( \prod_{x\in A} S_{x}^{1} \Bigr)_{-} (-H_{\Lambda,+})^{k} \Bigl( \prod_{x\in B} S_{x}^{1} \Bigr)_{-} (-H_{\Lambda,+})^{\ell}
\ee
with $k,\ell \in \bbN$. Expanding $(-H_{\Lambda,+})^k$ and $(-H_{\Lambda,+})^\ell$, we get a positive linear combination of
\be
\Tr_{\caH_\Lambda \otimes \caH_\Lambda} \Bigl( \prod_{x\in A} S_{x}^{1} \Bigr)_{-} \prod_{i=1}^k \Bigl( \prod_{x\in A_i} S_x^{\eps_i} \Bigr)_+ \Bigl( \prod_{x\in B} S_{x}^{1} \Bigr)_{-} \prod_{j=1}^\ell \Bigl( \prod_{x\in A_j'} S_x^{\eps_j'} \Bigr)_+
\ee
with $\eps_i, \eps_j' \in \{1,3\}$. Further, all products $(\prod S_{x}^{i})_{\pm}$ can be expanded using Lemma \ref{lem useful} in polynomials of $S_{x,\pm}^{i}$, still with positive coefficients. Finally, observe that the total number of operators $S_{x,-}^3$, $x \in \Lambda$, is always even; then each $S_{x,-}^{3}$ can be replaced by $-S_{x,-}^{3}$. We now have the trace of a polynomial, with positive coefficients, of matrices with nonnegative elements (by Lemma \ref{lem key}). This is positive.

The second inequality (with $S^{3}$ instead of $S^{2}$) is similar. The only difference is that $(\prod S_{x}^{3})_{-}$ gives a polynomial where the number of $S_{x,-}^{3}$ is odd. Hence the negative sign.
\end{proof}

\section{The case $S=1$}
\label{sec S=1}

This section is much more involved, and our result is sadly restricted to the ground state.
Our strategy is inspired by the work of Nachtergaele on graphical representations of the Heisenberg model with large spins \cite{Nac}. We consider a system where each site hosts a pair of spin $\frac12$ particles. The inequalities of Theorem \ref{thm ineq} apply. By projecting onto the triplet subspaces, one gets a correspondence with the original spin 1 system. We prove that all ground states of the new model lie in the triplet subspace, so the inequality can be transferred. These steps are detailed in the rest of the section.

It is perhaps worth noticing that the tensor products in this section play a different r\^ole than those in Section \ref{sec S=1/2}.

\subsection{The new model}

We introduce the new lattice $\tilde\Lambda = \Lambda \times \{1,2\}$. The new Hilbert space is
\be
\tilde\caH_{\Lambda} = \otimes_{x\in\Lambda} (\bbC^{2} \otimes \bbC^{2}) \simeq \otimes_{x \in \tilde\Lambda} \bbC^{2}.
\ee
Let $R^i$ be the following operator on $\bbC^2 \otimes \bbC^2$:
\be
\label{R}
 R^i = \tfrac{1}{2}(\sigma^i\otimes\mathbbm{1}+\mathbbm{1}\otimes\sigma^i).
\ee
Here, $\sigma^{i}$ are the Pauli matrices in $\bbC^{2}$ as before.
We denote $R_{x}^{i} = R^{i} \otimes \bbone_{\Lambda \setminus \{x\}}$ the corresponding operator at site $x \in \Lambda$. As before, we choose the interactions to be in the 1-3 spin directions; the hamiltonian on $\tilde\caH_{\Lambda}$ is
\be
\tilde{H}_{\Lambda} = -\sum_{A\subset\Lambda} \Bigl( J^1_A\prod_{x\in A}R^1_x+J^3_A\prod_{x\in A}R^3_x \Bigr).
\ee
The coupling constants $J_{A}^{i}$ are the same as those of the original model on $\caH_{\Lambda}$. 
The expected value of an observable $a$ in the Gibbs state with hamiltonian  $\tilde{H}_{\Lambda}$ is
\be
\langle a \rangle^{\sim} = \frac{1}{\tilde{Z}(\Lambda)} \Tr a \e{-\beta \tilde{H}_{\Lambda}},
\ee
where the normalisation $\tilde{Z}(\Lambda)$ is the partition function
\be
\tilde{Z}(\Lambda) = \Tr \e{-\beta \tilde{H}_{\Lambda}}.
\ee
We similarly define Schwinger functions $\langle \cdot; \cdot \rangle^{\sim}_{s}$ for $s \in [0,1]$.

It is useful to rewrite $\tilde H_\Lambda$ as the hamiltonian of spin $\frac12$ particles on the extended lattice $\tilde\Lambda$. Given a subset $X \subset \tilde\Lambda$, we denote ${\rm supp} X$ its natural projection onto $\Lambda$, i.e.
\be
{\rm supp} X = \bigl\{ x \in \Lambda : (x,1) \in X \text{ or } (x,2) \in X \bigr\}.
\ee
We also denote $D(\tilde\Lambda)$ the family of subsets of $\tilde\Lambda$ where each site of $\Lambda$ appears at most once. Notice that $|D(\tilde\Lambda)| = 3^{|\Lambda|}$. Finally, let us introduce the coupling constants
\be
\tilde J_{X}^{i} = \begin{cases} 2^{-|X|} \, J^{i}_{{\rm supp} X} & \text{if } X \in D(\tilde\Lambda), \\ 0 & \text{otherwise.} \end{cases}
\ee
From these definitions, we can write $\tilde H_{\Lambda}$ using Pauli operators as
\be
\label{ham new look}
\tilde{H}_{\Lambda} = -\sum_{X\subset\tilde{\Lambda}} \Bigl( \tilde{J}^1_X \prod_{x\in X} \sigma^1_x + \tilde{J}^3_X\prod_{x\in X} \sigma^3_x \Bigr).
\ee

\subsection{Correspondence with the spin 1 model}

The Hilbert space at a given site, $\bbC^2 \otimes \bbC^2$, is the orthogonal sum of the triplet subspace (that is, the symmetric subspace, which is of dimension 3) and of the singlet subspace (of dimension 1). Let $P^{\rm triplet}$ denote the projector onto the triplet subspace, and let $P_\Lambda^{\rm triplet} = \otimes_{x\in\Lambda} P^{\rm triplet}$. We define a new Gibbs state, namely
\be
\langle a \rangle' = \frac{1}{Z'(\Lambda)} \Tr a P_\Lambda^{\rm triplet} \e{-\beta \tilde{H}_{\Lambda}},
\ee
with partition function $Z'(\Lambda) = \Tr P_\Lambda^{\rm triplet} \e{-\beta \tilde{H}_{\Lambda}}$. In order to state the correspondence between the models with different spins, let 
 $V : \bbC^{3} \to \bbC^{2} \otimes \bbC^{2}$ denote an isometry such that
\be
\begin{split}
& V^{*} V = \bbone_{\bbC^{3}}, \\
& V V^{*} = P^{\rm triplet},
\end{split}
\ee
One can check that $S^i = V^* R^i V$, $i=1,2,3$, give spin operators in $\bbC^3$. Let $V_\Lambda = \otimes_{x\in\Lambda} V$. For all observables on $a \in \caH_\Lambda$, we have the identity
\be\label{newgibbs}
\langle a \rangle = \langle V_\Lambda a V_\Lambda^*\rangle'.
\ee

\subsection{All ground states lie in the triplet subspace}

Let $Q_{\Lambda,A}$ be the projector onto triplets on $A$, and singlet on $\Lambda\setminus A$:
\be
Q_{\Lambda,A} = \Bigl( \otimes_{x\in A} P^{\rm triplet} \Bigr) \otimes \Bigl( \otimes_{x\in\Lambda\setminus A} (1-P^{\rm triplet}) \Bigr).
\ee
One can check that $[R^i_x, Q_{\Lambda,A}]=0$ for all $i=1,2,3$, all $x\in\Lambda$, and all $A \subset\Lambda$. Further, since the operators $R^i$ give zero when applied on singlets, we have
\be
\label{ham in subspace}
Q_{\Lambda,A} \, \tilde H_\Lambda = -Q_{\Lambda,A} \sum_{B \subset A} \Bigl( J^1_B \prod_{x\in B} R^1_x+J^3_B \prod_{x\in B} R^3_x \Bigr).
\ee

\begin{lemma}
\label{lem gs}
The ground state energy of $\tilde H_\Lambda$ is a strictly decreasing function of $J_A^i$, for all $i=1,3$ and all $A \subset \Lambda$.
\end{lemma}

It follows from this lemma and Eq.\ \eqref{ham in subspace} that all ground states lie in the subspace of $Q_{\Lambda,\Lambda} = P_\Lambda^{\rm triplet}$. To see this, note that $Q_{\Lambda,A}$ has the effect of setting $J^i_B=0$ for $B\nsubseteq A$. Then if $A'\supset A$, $Q_{\Lambda,A'}\tilde H_\Lambda$ has larger coupling constants for those $B$ such that $B\nsubseteq A$ but $B\subset A'$ (other coupling constants are unaffected). Hence having more sites in the triplet subspace leads to strictly lower energy.

\begin{proof}[Proof of Lemma \ref{lem gs}]
We actually prove the result for the hamiltonian \eqref{ham new look} and the couplings $\tilde J_X^i$, which implies the lemma.
With $E_0(a)$ denoting the ground state energy of the operator $a$, we show that
\be
\label{this is enough}
E_0 \Bigl( \tilde H_\Lambda - \eps \prod_{x \in Y} \sigma_x^1 \Bigr) < E_0(\tilde H_\Lambda),
\ee
for any $\eps>0$, and any $Y\subset\tilde\Lambda$. Let $\psi_0$ denote the ground state of $\tilde H_\Lambda$. It is also eigenstate of $\e{-\tilde H_\Lambda}$ with the largest eigenvalue. Using the Trotter product formula, we have
\be
\e{-\tilde H_\Lambda} = \lim_{n\to\infty} \biggl[ \Bigl( 1 + \frac1n \sum_{X \subset \Lambda} \tilde J_X^1 \prod_{x\in X} \sigma_x^1 \Bigr) \e{\frac1n \sum_{X \subset \Lambda} \tilde J_X^3 \prod_{x\in X} \sigma_x^3} \biggr]^n,
\ee
which, in the basis where the Pauli matrices are given by \eqref{def Pauli}, is a product of matrices with nonnegative elements. By a Perron-Frobenius argument, $\psi_0$ can be chosen as a linear combination of the basis vectors with nonnegative coefficients. Then
\be
E_0 \Bigl( \tilde H_\Lambda - \eps \prod_{x \in Y} \sigma_x^1 \Bigr) \leq \Bigl( \psi_0, \Bigl( \tilde H_\Lambda - \eps \prod_{x \in Y} \sigma_x^1 \Bigr) \psi_0 \Bigr) = E_0(\tilde H_\Lambda) - \eps \Bigl( \psi_0, \Bigl( \prod_{x \in Y} \sigma_x^1 \Bigr) \psi_0 \Bigr).
\ee
If $(\psi_0, (\prod_{x \in Y} \sigma_x^1) \psi_0) \neq 0$, then it is positive and the conclusion follows. Otherwise, let $\hat H_\Lambda = \tilde H_\Lambda - c \bbone$ with $c$ large enough so that all eigenvalues of $\hat H_\Lambda$ are negative. If $(\psi_0, (\prod_{x \in Y} \sigma_x^1) \psi_0) = 0$, using $\bigl( \prod_{x\in Y} \sigma^1_x \bigr)^2 = 4^{-|Y|} \bbone$, we have
\be
\Bigl( \psi_0, \Bigl( \hat H_\Lambda - \eps \prod_{x \in Y} \sigma_x^1 \Bigr)^2 \psi_0 \Bigr) = E_0(\hat H_\Lambda)^2 + (4^{-|Y|} \eps)^2.
\ee
This implies that $E_0(\hat H_\Lambda - \eps \prod \sigma_x^1) < E_0(\hat H_\Lambda)$, hence the strict inequality \eqref{this is enough}.

One can replace $\sigma_x^1$ with $\sigma_x^3$ and prove Inequality \eqref{this is enough} in the same fashion; indeed, one can choose a basis where $\sigma^3$ is like $\sigma^1$ in \eqref{def Pauli}, and $\sigma^1$ is like $-\sigma^3$.
\end{proof}

\subsection{Proof of Theorem \ref{thm ineq2}}

We can assume that for any $x \in \Lambda$, there exist $i \in \{1,3\}$ and $A \ni x$ such that $J_A^i > 0$ --- the extension to the general case is straightforward. Since all ground states lie in the triplet subspace, we have for all subsets $A,B\subset \Lambda$ and for all $s \in [0,1]$,
\be
\begin{split}
\lim_{\beta\to\infty} \Bigl\langle \prod_{x\in A}S_x^1; \prod_{x\in B}S_x^1 \Bigr\rangle_{s} &= 2^{-|A|-|B|} \sumtwo{X \in D(\tilde\Lambda)}{{\rm supp}X = A} \; \sumtwo{Y \in D(\tilde\Lambda)}{{\rm supp}Y = B} \lim_{\beta\to\infty} \Bigl\langle \prod_{x\in X}\sigma_x^1; \prod_{x\in Y} \sigma_x^1 \Bigr\rangle^{\sim}_{s} 
 \\
&\geq 2^{-|A|-|B|} \sumtwo{X \in D(\tilde\Lambda)}{{\rm supp}X = A} \; \sumtwo{Y \in D(\tilde\Lambda)}{{\rm supp}Y = B} \lim_{\beta\to\infty} \Bigl\langle \prod_{x\in X}\sigma_x^1 \Bigr\rangle^{\sim} \; \Bigl\langle\prod_{x\in Y} \sigma_x^1 \Bigr\rangle^{\sim} \\
 &= \lim_{\beta\to\infty} \Bigl\langle \prod_{x\in A} S_x^1 \Bigr\rangle \; \Bigl\langle \prod_{x\in B} S_x^1 \Bigr\rangle.
\end{split}
\ee
We used Theorem \ref{thm ineq}. We have obtained the first inequality of Theorem \ref{thm ineq2}. The second inequality follows in the same way.

\medskip
\noindent
{\bf Acknowledgments:} DU is grateful to J\"urg Fr\"ohlich, Charles-\'Edouard  Pfister, and Yvan Velenik, who taught him much about correlation inequalities. We thank the referees for useful comments. DU thanks the Erwin Schr\"odinger Institute for a useful visit during the program ``Quantum many-body systems, random matrices, and disorder'' in July 2015. BL is supported by the EPSRC grant EP/HO23364/1.

{
\renewcommand{\refname}{\small References}
\bibliographystyle{symposium}

}

\end{document}